\documentclass[aps, reprint, amsmath,amssymb, prl, superscriptaddress]{revtex4-1}
\usepackage{svg}
\usepackage{color}
\usepackage{graphicx}
\usepackage{dcolumn}
\usepackage{bm}
\usepackage{physics}
\usepackage{hyperref}
\usepackage{url}
\usepackage{amsmath}
\usepackage{tikz-cd}
\usepackage{notes2bib}
\usepackage{ dsfont }
\usepackage{svg}
\usepackage{verbatim}
\usepackage{amssymb}
\usepackage{amsfonts}              
\usepackage{amsthm}                
\usepackage{mathtools}
\usepackage{qcircuit}
\usepackage[normalem]{ulem}
\usepackage{xcolor}

\DeclarePairedDelimiterX{\barpair}[2]{(}{)}{%
  #1\;\delimsize\|\;#2%
}
\theoremstyle{definition}
\newtheorem{theorem}{Theorem}
\newtheorem{prop}[theorem]{Proposition}

\newtheorem{lemma}[theorem]{Lemma}

\newtheorem{corollary}[theorem]{Corollary}

\newcommand{\Ad}{\text{Ad}}
\newcommand{\id}{\text{id}}
\newcommand{\mf}[1]{\mathfrak{#1}}
\newcommand{\mcal}[1]{\mathcal{#1}}

\newcommand{\jami}{Jamio{\l}kowski }

\newcommand{\mds}[1]{\mathds{#1}}
\newcommand{\supp}[1]{\text{supp}\left(#1\right)}

\newcommand{\zpan}[1]{\text{zpan}_\mds{R}\left(#1\right)}
\newcommand{\seepr}{Proof can be found in Appendix. }
\newcommand{\seeap}{See Appendix for proof. }

\bibnotesetup{note-name    = ,use-sort-key = false} 
\begin{document}

\title{Faithfulness and sensitivity for ancilla-assisted process tomography }

\author{Seok Hyung Lie}
\affiliation{%
 Department of Physics and Astronomy, Seoul National University, Seoul, 151-742, Korea
}%
\affiliation{%
 School of Physical and Mathematical Sciences, Nanyang Technological University, 21 Nanyang Link, Singapore, 637371
}%
\author{Hyunseok Jeong}
\affiliation{%
 Department of Physics and Astronomy, Seoul National University, Seoul, 151-742, Korea
}%

\date{\today}

\begin{abstract}
 A system-ancilla bipartite state capable of containing the complete information of an unknown quantum channel acting on the system is called faithful. The equivalence between faithfulness of state and invertibility of the corresponding Jamio{\l}kowski map proved by D'Ariano and Presti has been a useful characterization for ancilla-assisted process tomography albeit the proof was incomplete as they assumed trace nonincreasing quantum operations, not quantum channels. We complete the proof of the equivalence and introduce the generalization of faithfulness to various classes of quantum channels. We also explore a more general notion we call sensitivity, the property of quantum state being altered by any nontrivial action of quantum channel. We study their relationship by characterizing both properties for important classes of quantum channels such as unital channels, random unitary operations and unitary operations. Unexpected (non-)equivalence results among them shed light on the structure of quantum channels by showing that we need only two classes of quantum states for characterizing quantum states faithful or sensitive to various subclasses of quantum channels. For example, it reveals the relation between quantum process tomography and quantum correlation as it turns out that only bipartite states that has no local classical observable at all can be used to sense the effect of unital channels.
\end{abstract}

\pacs{Valid PACS appear here}
\maketitle

\textit{{Introduction}---\,}Quantum process tomography (QPT) is a task of identifying an unknown quantum process by feeding it with a set of known input states and measuring the corresponding outcomes. For QPT to be successful, the set of input states $\{\rho_n\}$ should be \textit{informationally complete}, i.e., for every quantum process $\mcal{E}$, the corresponding set of output states $\{\mcal{E}(\rho_n)\}$ should be unique. In general, multiple types of input states are required to uniquely identify a quantum process, but by employing the ancilla-assisted process tomography (AAPT) \cite{altepeter2003ancilla,mohseni2008quantum,shukla2014single,teo2020objective,o2004quantum,howard2006quantum,d2004quantum}, it is possible to identify a quantum channel by using copies of only one type of input state albeit it requires correlation between the system and the environment. (See FIG. \ref{fig:aapt}.) In Ref. \cite{d2003imprinting}, when a bipartite quantum system transforms to a different state for every quantum process applied to a part of its systems, it is called \textit{faithful}. In other words, a bipartite quantum state $\rho_{AB}$ is faithful if the mapping from quantum channel $\mcal{E}$ to $\mcal{E}_A(\rho_{AB})$ is injective (one-to-one).

On the other hand, for every bipartite state $\rho_{AB}$, one could define its corresponding \jami map $\rho_{A\to B}$ given as
\begin{equation} \label{eqn:jami}
    \rho_{A\to B}(\sigma)=\Tr_A[(\sigma_A^T\otimes\mds{1}_B)\rho_{AB}].
\end{equation}
It was claimed in Ref. \cite{d2003imprinting} that the invertibility of $\rho_{A\to B}$ and the faithfulness of $\rho_{AB}$ are equivalent. Just as the Choi-\jami isomorphism \cite{choi1975completely,jamiolkowski1972linear} relates static and dynamic pictures of quantum processes, this equivalence is significant as it relates quantum state tomography with quantum process tomography. Moreover. faithfulness is an intuitive method for showing that even mixed states can be faithful.

However, the argument given in Ref. \cite{d2003imprinting} for the equivalence of the faithfulness and the invertibility works only when the faithfulness is assumed to be about quantum operations, not quantum channels. A quantum operation $\mcal{O}(\rho)=\sum_i K_i \rho K_i^\dag$ is a trace non-increasing CP (completely-positive) map, (equivalent to $\sum_i K_i^\dag K_i \leq \mds{1}$), whereas a quantum channel $\mcal{C}(\rho)=\sum_i L_i \rho L_i^\dag$ is a trace-preserving CP map (equivalent to $\sum_i L_i^\dag L_i =\mds{1}$). However, in general, a quantum operation can be implemented, even nondeterministically, only when some information of the environment discarded in the implementation of the map is revealed. It cannot be always justified, hence the assumption of faithfulness to general quantum operations is sometimes too strong, and it is desirable to require faithfulness only for quantum channels.

\begin{figure}[t]
    \includegraphics[width=.4\textwidth]{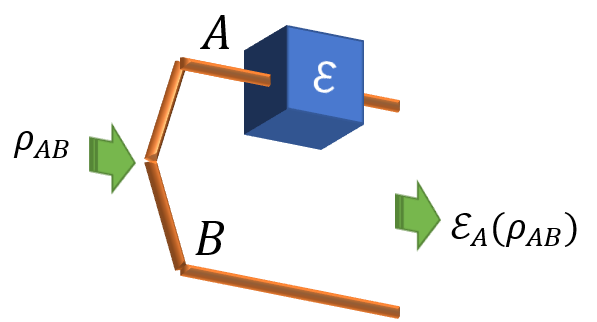}
    \caption{Ancilla-assisted process tomography. Only a part of correlated quantum state $\rho_{AB}$ passes through an unknown channel $\mcal{E}$, but if the state $\rho_{AB}$ is \textit{faithful}, then the channel tomography can be done by the state tomography of $\mcal{E}_A(\rho_{AB})$.}\label{fig:aapt}
\end{figure}

The main problem of the argument in Ref. \cite{d2003imprinting} is that the span of the set of quantum channels does not form the whole space of linear maps whereas the set of quantum operations does, which is required to directly derive the equivalence. Fortunately, we will show that the result does not change even when one requires the faithfulness to hold only for quantum channels. In this work, we provide a complete proof of the original claim of the equivalence of bipartite quantum channel and invertibility of its corresponding \jami map.

Additionally, we study a new property we will call \textit{sensitivity} of bipartite state, a property that requires the state to change whenever a nontrivial quantum process is applied to a part of its systems. This property is useful to study catalysts in quantum resources theories which must not be changed by reaction of resource transformation. We explore the equivalence relations between faithfulness and sensitivity of quantum states to various classes of quantum channels such as unital maps, random unitary maps and unitary operations and their relation to the properties of the corresponding \jami map. 

Identifying the exact condition for being tomographically faithful or sensitive conversely defines the set of states that are \textit{not} faithful or sensitive. It will serve as a cornerstone of establishing the resource theory of tomographical power because they form the set of free states, which is the central part of a resource theory \cite{chitambar2019quantum}.

\textit{Notation}---\,We say a function $f:S\to T$ is injective (one-to-one or left invertible) when $f(x)=f(y)$ implies $x=y$ for all $x,y\in S$ and surjective (onto or right invertible) when for all $y\in T$, there exists $x\in S$ such that $y=f(x)$. Without loss of generality, we sometimes identify the Hilbert space $H_X$ corresponding to a quantum system $X$ with the system itself and use the same symbol $X$ to denote both. We will denote the dimension of $X$ by $|X|$.

The space of all operators acting on system $X$ is denoted by $\mf{B}(X)$, the real space of all Hermitian matrices on system $X$ by $\mf{H}(X)$. The set of all unitary operators in $\mf{B}(X)$ is denoted by $\mf{U}(X)$. For any matrix $M$, $M^T$ is its transpose with respect to some fixed basis. We treat $\mf{B}(X)$ as a Hilbert space equipped with the Hilbert-Schmidt inner product $\langle A,B\rangle :=\Tr[A^\dag B]$.

The space of all linear maps from $\mf{B}(X)$ to $\mf{B}(Y)$ is denoted by $\mf{L}(X, Y)=\mf{B}(\mf{B}(X),\mf{B}(Y))$ and we will used the shorthand notation $\mf{L}(X):=\mf{L}(X, X)$. The set of all quantum channels (completely positive and trace-preserving linear maps) from system $X$ to $Y$ is denoted by $\mf{C}(X , Y)$ with $\mf{C}(X):=\mf{C}(X , X)$. For any linear map $\mcal{N}$ we define its transpose $\mcal{N}^T$ as $\mcal{N}^T(\rho):=\mcal{N}^\dag(\rho^T)^T$. We let the set of unitary operations on $X$ be $\mf{UO}(X)$ and call the channels in the convex hull of $\mf{RU}(X):=\text{conv}(\mf{UO}(X))$ random unitary operations, and let the set of unital quantum channels be $\mf{UC}(X):=\{\mcal{N}\in\mf{C}(X):\mcal{N}(\mds{1}_X)=\mds{1}_X\}$.  We denote the identity map on system $X$ by $\id_X$.

\textit{Tomographic faithfulness}---
 The argument for the equivalence between the faithfulness and the invertibility of Jamio{\l}kowski map given in Ref. \cite{d2003imprinting} is incomplete because the set of quantum channels does not span the whole complex vector space of linear maps. We give a complete proof of that claim here. Before getting into faithfulness of bipartite states, we first investigate the necessary and sufficient condition for injectivity and surjectivity of quantum channels. We first show the equivalence with the invertibility on the space of traceless Hermitian operators. We only work with finite dimensional systems in this work.
 
  To prove our first main result, we need some observation on the structure of the set of linear maps of traceless Hermitian operators that connects trace non-preserving quantum operations and trace-preserving quantum channels.
 
 \begin{lemma} \label{applem:1}
    Any linear map on the space of traceless Hermitian operators can be expressed as a scalar multiple of the difference of two quantum channels.
\end{lemma}
 
Combining it with other technical Lemmas introduced in Appendix, we can prove one of the main results, which is a slight generalization of the result first implicitly introduced in Ref. \cite{d2003imprinting} with an incomplete argument.

\begin{prop} \label{thm:dariano}
    For any $\mcal{R}\in \mf{C}(X,Y)$,  $\mcal{R}\circ\mcal{N}=\mcal{R}\circ\mcal{M}$ implies $\mcal{N}=\mcal{M}$ for every $\mcal{N},\mcal{M}\in \mf{C}(X)$ (or if $\mcal{N}\circ\mcal{R}=\mcal{M}\circ\mcal{R}$ implies $\mcal{N}=\mcal{M}$ for every $\mcal{N},\mcal{M}\in \mf{C}(Y)$) if and only if $\mcal{R}$ is injective (or surjective).
\end{prop}
See Appendix for omitted proofs. Unfortunately, although the \jami map of a general bipartite state is a CP maps, it need not be a quantum channel. To apply this Theorem to general CP maps that need not preserve trace, we derive the following result.
\begin{corollary} \label{appco:1}
    When $\mcal{R}\in\mf{L}(X,Y)$ is a CP map such that $\mcal{R}(\rho) \neq 0$ for any nonzero $\rho\geq0$, $\mcal{N}\circ\mcal{R}=\mcal{M}\circ\mcal{R}$ implies $\mcal{N}=\mcal{M}$ for every $\mcal{N},\mcal{M}\in \mf{C}(Y)$ if and only if $\mcal{R}$ is surjective.
\end{corollary}

The following basic result from linear algebra will be helpful when we switch between the perspectives of $A$ and $B$. \seepr

\begin{lemma} \label{applem:5}
    $\mcal{R}\in\mf{B}(X,Y)$ is surjective if and only if $\mcal{R}^\dag (\text{or } \mcal{R}^T) \in\mf{B}(Y,X)$ is injective.
\end{lemma}

Consider a bipartite state $\rho_{AB}$. Here, without loss of generality, we assume that $A=\supp{\rho_A}$ and $B=\supp{\rho_B}$ so that both of $\rho_A$ and $\rho_B$ are full rank. It means that for any nonzero $\sigma_B\geq0$ in $\mf{B}(B)$, $\Tr[\sigma_B\rho_B]>0$. We say that $\rho_{AB}$ is \textit{faithful} (for quantum process tomography) on $A$, when the mapping $\mcal{E}\mapsto \mcal{E}_A(\rho_{AB})$ defined for any quantum channel $\mcal{E}\in\mf{C}(A)$ is injective. If this property holds for all $\mcal{E}\in\mf{Q}$ for some other subset $\mf{Q}\in\mf{L}(A)$ instead, then we will say that $\rho_{AB}$ is faithful to $\mf{Q}$. Faithfulness is crucial for quantum process tomography, since one can always recover the unknown quantum operation acted on a faithful quantum state by implementing quantum state tomography.

We claim that $\rho_{AB}$ being faithful on $A$ is equivalent to that the Jamio{\l}kowski map $\rho_{B\to A}\in \mf{L}(B,A)$ defined as
\begin{equation}
    \rho_{B\to A}(\sigma) = \Tr_B[(\mds{1}_A\otimes\sigma_B^T)\rho_{AB}],
\end{equation}
is surjective. It is because, first, $\mcal{E}\mapsto \mcal{E}_A(\rho_{AB})$ being injective is equivalent to that $\mcal{N}\circ\rho_{B\to A}=\mcal{M}\circ\rho_{B\to A}$ implies $\mcal{N}=\mcal{M}$ by the facts that the mapping $\rho_{AB} \mapsto \rho_{B\to A}$ is bijective and that $(\mcal{E}_A(\rho_{AB}))_{B\to A}=\mcal{E}\circ \rho_{B\to A}$. Next, by Corollary \ref{appco:1}, it is in turn equivalent to $\rho_{B\to A}$ being surjective since $\Tr[\rho_{B\to A}(\sigma)]=\Tr[\sigma^T\rho_B]>0$ for every nonzero $\sigma\geq0$.

Now, from the definition of $\rho_{A\to B}\in\mf{L}(A,B)$ as in (\ref{eqn:jami}), we can observe that $\rho_{A\to B}=\rho_{B\to A}^T$. Therefore, by Lemma \ref{applem:5}, it follows that $\rho_{AB}$ is faithful on $A$ if and only if $\rho_{A\to B}$ is injective, i.e. left invertible.

\begin{theorem} \label{thm:dariano2}
    A bipartite state $\rho_{AB}$ is faithful for quantum process tomography on $A$ if and only if its Jamio{\l}kowski map $\rho_{A\to B}$ is left invertible.
\end{theorem}

Next, we consider faithfulness to a smaller class of quantum channels since it is often the case that an unknown channel can be sampled from a more restricted set of quantum channels.
In this work, we focus on the set of unital channels and its subclasses. Unital channels are important as they are characterized by the property of never decreasing the entropy of input state, so that they can model randomizing noises affecting quantum systems.

Unital channels are also called bistochastic maps because of its formal similarity with bistochastic matrices of classical probability theory. The set of random unitary operations is an important subclass of unital channels since they describe the noises originated from classical randomness so that  it can be considered the quantum counterpart of random permutation operations \cite{boes2018catalytic,lie2020uniform,lie2021correlational}. However, a remarkable result is that the quantum Birkhoff-von Neumann theorem is not true, i.e., although every bistochastic matrix is a convex sum of permutation matrices, not every unital channel is a random unitary operation, or a quantum  channel that can be expressed a convex sum of unitary operations. Despite the gap between these classes of quantum channels, we have the following surprising result.

\begin{theorem} \label{thm:equiv}
    Faithfulness to quantum channels, to unital quantum channels, and to random unitary operations are all equivalent.
\end{theorem}

\seepr However, this reduction of faithfulness has a certain limitation as the following result suggests. Namely, faithfulness to quantum channels and to unitary operations are not equivalent. \seeap 

\begin{prop} \label{prop:unitinv}

There is a bipartite quantum state $\rho_{AB}$ that is faithful to unitary operations on $A$ but whose \jami map $\rho_{A\to B}$ is not left invertible. 
\end{prop}

\textit{Tomographic sensitivity}---\, In this section, we will consider a property of bipartite quantum state that generalizes the faithfulness. Oftentimes, the full tomography of unknown channel is not necessary when one is only interested in \textit{sensing} the perturbation caused by the quantum channel. In that case, a quantum state's ability of detecting presence of nontrivial effect is more important than that of faithfully remembering every detail of the channel action.

We will say that a bipartite state $\rho_{AB}$ is (tomogrphically) \textit{sensitive} on $A$ when, for every $\mcal{E}\in\mf{C}(A)$, $\mcal{E}_A(\rho_{AB})=\rho_{AB}$ implies $\mcal{E}=\id_A$. If this property holds for all $\mcal{E}\in\mf{Q}$ for some other subset $\mf{Q}\subseteq\mf{L}(A)$ instead, then we will say that $\rho_{AB}$ is sensitive to $\mf{Q}$. A quantum state is faithful when it is useful for fully recovering an unknown quantum process, whereas a quantum state is sensitive for merely detecting if a change has happened to the state. From definition, we immediately get the following relation between faithfulness and sensitivity.

\begin{prop}
    Faithfulness to $\mf{Q}$ and sensitivity to $\mf{Q}$ are equivalent when $\mf{Q}$ forms a group.
\end{prop}
\begin{proof}
    By definition, every faithful state is sensitive for any $\mf{Q}$. Conversely, assume that $\mf{Q}$ is a group. It follows that for every $\mcal{F}\in\mf{Q}$, its inverse $\mcal{F}^{-1}$ exists in $\mf{Q}$. Therefore, the condition that $\mcal{E}_A(\rho_{AB})=\mcal{F}_A(\rho_{AB})$ implies $\mcal{E}=\mcal{F}$ for arbitrary $\mcal{E},\mcal{F}\in\mf{Q}$ is equivalent to that $(\mcal{F}^{-1}\mcal{E})_A(\rho_{AB})=\rho_{AB}$ implies $\mcal{F}^{-1}\mcal{E}=\id_A$. Since $\mcal{F}^{-1}\mcal{E}$ is still arbitrary, the desired result follows.
\end{proof}

We remark that when the input and output spaces dimensions are the smae, then the only possible $\mf{Q}$ that forms a group is the set of unitary operations or its subgroup. Thus, faithfulness and sensitivity to unitary operations are equivalent. However, there are cases where faithfulness and sensitivity are distinct, as given in Appendix.

\begin{prop} \label{prop:nonequiv}
    Faithfulness and sensitivity to unital channels are not equivalent. The same holds for random unitary operations.
\end{prop}

Recall that $\rho_{AB}$ is called a C-Q(Classical-quantum) state if there is a basis $\{\ket{i}_A\}$ of $A$ such that $\rho_{AB}$ has the form
    \begin{equation} \label{eqn:CQ}
        \rho_{AB}=\sum_i p_i \dyad{i}_A\otimes \rho_B^i,
    \end{equation}
for some probability distribution $\{p_i\}$ and quantum states $\rho_B^i\in\mf{S}(B)$. As a generalization, we will say that $\rho_{AB}$ is PC-Q (partially classical-quantum) if there exists a projective measurement $\{\Pi_i\}_{i=1}^n$ with $n>1$ on $A$ ($\Pi_i\Pi_j=\delta_{ij}\Pi_i$ and $\sum_i \Pi_i =\mds{1}_A$) that leaves $\rho_{AB}$ unperturbed. In other words,
    \begin{equation} \label{eqn:PCQ}
        \rho_{AB}=\sum_i(\Pi_i\otimes \mds{1}_B)\rho_{AB}(\Pi_i\otimes \mds{1}_B).
    \end{equation}
 If the roles of $A$ and $B$ are reversed, then we will call it a Q-PC state. A state $\rho_{AB}$ being not PC-Q means that there is no classical observable that can be observed without perturbing the state does not exist \textit{at all}. Hence, it signifies the existence of a certain type of quantum correlation between $A$ and $B$. Hereby we provide a unifying characterization of quantum states that are sensitive to unital channels and unitary operations.
 
\begin{theorem} \label{thm:PCQcond} The following are equivalent. $(i)$ $\rho_{AB}$ is sensitive to unitary operations on $A$. $(ii)$ $\rho_{AB}$ is sensitive to unital channels on $A$. $(iii)$ $\rho_{AB}$ is not PC-Q. 
\end{theorem}

\seepr Hence, for any set $\mf{W}$ of quantum channels on $A$ between unital channels and unitary operations $(\mf{UO}(A)\subseteq\mf{W}\subseteq\mf{UC}(A))$, sensitivity to $\mf{W}$ is equivalent to not being a PC-Q state. Important examples of such $\mf{W}$ include the set of random unitary operations, $\mf{RU}(A)$, the maps with the asymptotic quantum Birkhoff property (AQBP), the Schur maps, the catalytic maps and factorizable maps and its strong variant \cite{lie2021correlational,haagerup2011factorization,musat2020non}. From these observations, we can see that a bipartite quantum state is sensitive to unital channels \textit{on both sides} if and only if there are no local projective measurements that leave the state unperturbed.

One can generalize the concept of sensitivity even further to arbitrary two subclasses of linear maps $\mf{A}\subseteq\mf{B}\subseteq\mf{L}(A)$: a bipartite state $\rho_{AB}$ is sensitive to $\mf{B}$ up to $\mf{A}$ if $\mcal{F}_A(\rho_{AB})=\rho_{AB}$ implies $\mcal{F}\in\mf{A}$ for any $\mcal{F}\in\mf{B}$. For example, one can consider mixedness-sensitivity by setting $\mf{A}=\mf{UO}(A)$ and $\mf{B}=\mf{UC}(A)$ as it was studied in Ref. \cite{mendl2009unital}. One can consider it as a nonlinear generalized quantum resource witnesses which only captures the linear separation of quantum resource from free objects.

This type of generalized sensitivity finds its application in the study of catalytic quantum randomness \cite{lie2021dynamical}. In recent results of quantum resource theory, catalysts, quantum systems that aid transformation of quantum states or channels while staying in the same state, are of intense interest. Many of such catalysts are in a multipartite state (e.g. entanglement catalyst), thus characterizing which local quantum operation does not fix a given multipartite state will help us understand the mechanism of quantum catalysis.

\begin{figure}[t]
    \includegraphics[width=.47\textwidth]{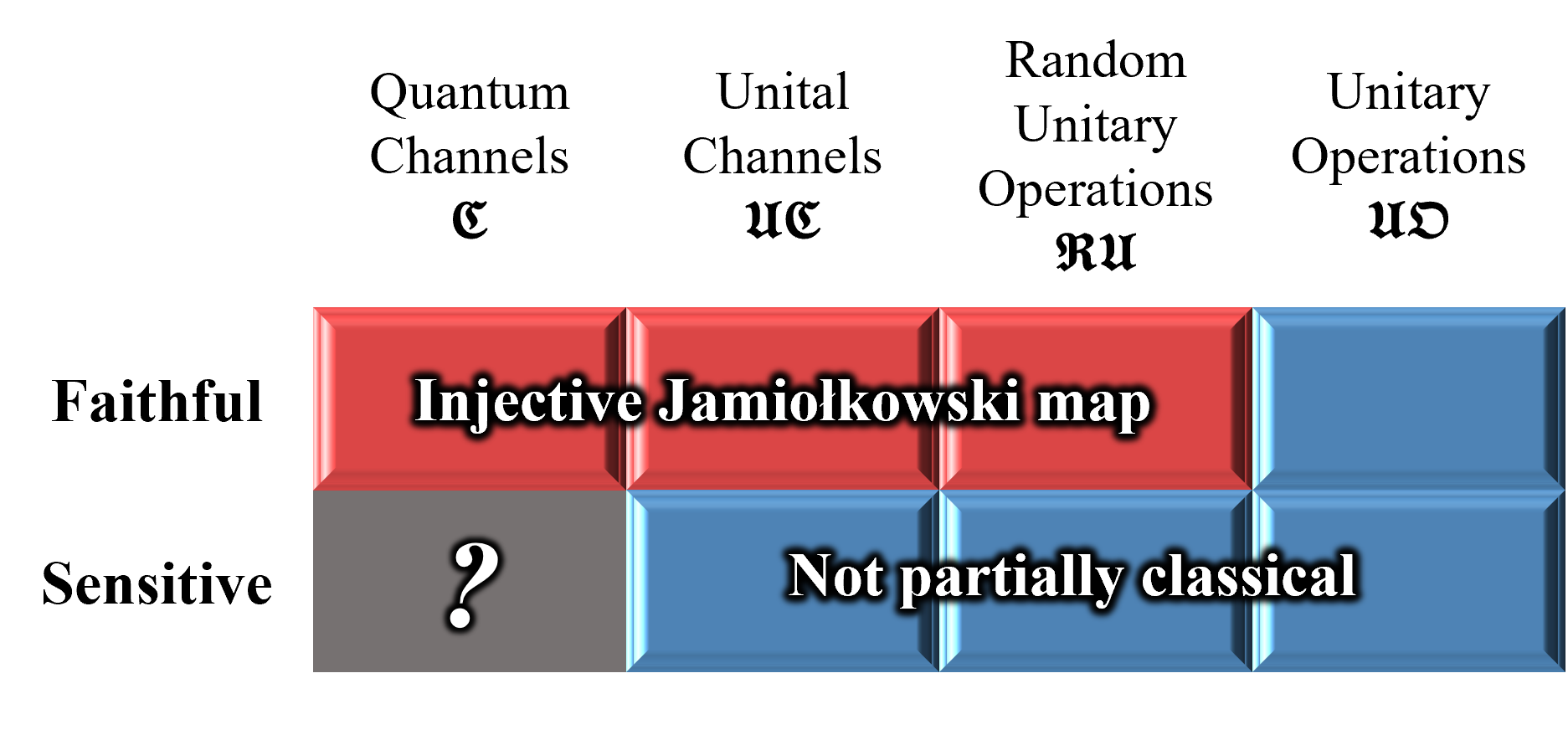}
    \caption{Equivalence of faithfulness and sensitivity to various subclasses of quantum channels. The characterization of sensitivity to general quantum channels is still unknown, but it should be logically in between of two other constraints given in the table.}\label{fig:table}
\end{figure}

\textit{Discussion}---\, We have completed the proof of a commonly used result in ancilla-assisted quantum process tomography, namely the characterization of tomographically faithful bipartite quantum states. In addition to this, we considered generalizations of faithfulness and its weaker variant sensitivity to various classes of quantum channels, and proved the equivalence between faithfulness to quantum channels, unital channels, and random unitary operations. As a result, we showed the equivalence of local nonclassicality (not being a PC-Q state) and sensitivity to unital channels.

The concepts of faithfulness and sensitivity are important even when one does not explicitly utilize reference systems in process tomography. The framework of AAPT unifies general strategies, hence it can even describe general process tomography methods that do not employ genuinely quantum ancillary systems. For example, the strategy of feeding an ensemble of input state states $\{\rho_i\}_{i=1}^{|A|^2}$ can be expressed as a single type of correlated bipartite input state $|A|^{-2}\sum_{i=1}^{|A|^2} \rho_i\otimes \dyad{i}_B$, where system $B$ functions as a classical register that remembers which input state was fed. Moreover, the equivalence of faithfulness and invertibility of corresponding \jami map even unifies both static and dynamic pictures of process tomography.

The characterization results also provide the characterization of states that are \textit{not} faithful or sensitive which are the free states in the resource theory where faithful or sensitive states are considered resourceful. The question of \textit{how} much the given state is faithful or sensitive can be studied in the constructed resource theory thereafter, and we anticipate the connection to be made with the field of quantum metrology in the future works.

It is rather surprising that faithfulness and sensitivity for many different classes of quantum channels can be expressed as one of only two mathematical characterizations as in FIG. \ref{fig:table}. However, the list is far from being complete. Notably the condition for sensitivity to quantum channels was not covered in this work, and it would be interesting to explore faithfulness and sensitivity of more various classes of quantum operations and find applications of theses results outside of process tomography.

\begin{acknowledgements}
We thank Giacomo D'Ariano for useful discussions. This work was supported by National Research Foundation of Korea grants funded by the Korea government (Grants No. 2019M3E4A1080074, No. 2020R1A2C1008609 and No. 2020K2A9A1A06102946) via the Institute of Applied Physics at Seoul National University and by the Ministry of Science and ICT, Korea, under the ITRC (Information Technology Research Center) support program (IITP-2020-0-01606) supervised by the IITP (Institute of Information \& Communications Technology Planning \& Evaluation). This work is also supported by the quantum computing technology development program of the National Research Foundation of Korea (NRF) funded by the Korean government (Ministry of Science and ICT (MSIT)) (No.2021M3H3A103657312).
\end{acknowledgements}
\bibliography{main}
\clearpage
\appendix
\widetext
\section{Mathematical proofs}
In addition to the notations introduced in the main text, we will use the following notations in this Appendix. For any $M\in\mf{B}(X\otimes Y)$, the partial transpose on system $X$ is denoted by $M^{T_X}$. For any $M\in \mf{B}(X,Y)$, we let $\Ad_M\in\mf{L}(X,Y)$ be $\Ad_M(K):=MKM^\dag$. The set of all quantum states on system $X$ will be denoted by $\mf{S}(X)$. Let $\mf{H}_0(X)$ be the space of all traceless Hermitian operators on Hilbert space $X$ and let $\mf{LH}(X,Y)$ be the space of Hermitian-preserving linear maps from $X$ to $Y$.

\subsection{Lemmas for Proposition \ref{thm:dariano}}
 For the construction of the proof of Proposition \ref{thm:dariano}, we need following key observations. First, we focus on the structure of the linear maps on the space of traceless Hermitian operators.

We first give a proof of Lemma \ref{applem:1}.
\begin{proof}
    Any Hermitian-preserving map has the expression $\mcal{C}$ of the form $\mcal{C}=\sum_i \lambda_i \Ad_{V_i}$ with real numbers $\lambda_i$ and some operators $\{V_i\}$ such that $\Tr[V_i^\dag V_j]=\delta_{ij}$. Note that we can set $\Tr[\mcal{C}(\rho)]=\Tr[\left(\sum_i\lambda_i V_i^\dag V_i\right)\rho]=0$ for every Hermitian $\rho$ by setting $\Tr[\mcal{C}(\mds{1})]=0$. It is possible because the value of $\mcal{C}(\mds{1})$ is not assigned as $\mcal{C}$ was defined only on traceless Hermitian operators. It follows that $\sum_i \lambda_i V_i^\dag V_i=0$. Let $S:=\{i:\lambda_i\geq0\}$. Let $M:=\sqrt{\alpha\mds{1}-\sum_{i\in S}\lambda_i V_i^\dag V_i}$ where $\alpha:=\|\sum_{i\in S} \lambda_i V_i^\dag V_i\|_\infty$ so that $\alpha\mds{1}\geq \sum_{i\in S}\lambda_i V_i^\dag V_i$. It follows that $\mcal{K}_0:=\alpha^{-1}\left(\sum_{i\in S} \lambda_i \Ad_{V_i}+\Ad_M\right)$ is a quantum channel for $\alpha^{-1}\sum_{i\in S}\lambda_i V_i^\dag V_i + \alpha^{-1} M^\dag M=\mds{1}$. Since $\sum_i \lambda_i V_i^\dag V_i=0$, we have that $\mcal{K}_1:=\alpha^{-1}\left(\sum_{i\notin S}|\lambda_i|\Ad_{V_i} + \Ad_M \right)$ is also a quantum channel. It follows that $\mcal{C}=\alpha\left(\mcal{K}_0-\mcal{K}_1\right)$.
\end{proof}

Interestingly, Lemma \ref{applem:1} holds even when we require the quantum channels to be invertible. If a linear map $\mcal{C}\in\mf{B}(\mf{H}_0(X))$ has an expression $\mcal{C}=\alpha(\mcal{K}_0-\mcal{K}_1)$ with some quantum channels $\mcal{K}_0$ and $\mcal{K}_1$, then there exists some $\epsilon>0$ such that both of $\mcal{K}'_i:=(1-\epsilon)\mcal{K}_i+\epsilon\;\id_X$ are invertible, thus we have $\mcal{C}=\alpha(1-\epsilon)^{-1}(\mcal{K}'_0-\mcal{K}'_1)$. Therefore, faithfulness is equivalent to faithfulness to invertible quantum channels. By the same argument, faithfulness to a convex subset of linear maps $\mf{Q}\in\mf{L}(X)$ with $\id_X\in\mf{Q}$ is equivalent to faithfulness of invertible linear maps in $\mf{Q}$.

Next, we remark that quantum channels are Hermitian-preserving maps, and the following Lemma helps us exploit the structure of Hermitian-preserving maps.
\begin{lemma} \label{applem:4}
    A Hermitian-preserving linear map that is injective (or surjective) on Hermitian operators is also injective (or surjective) on every bounded operators.
\end{lemma}
\begin{proof}
    We remark that every $A\in\mf{B}(X)$ can be decomposed into $A_R+iA_I$ where $A_R=(A+A^\dag)/2$ and $A_I=-i(A-A^\dag)/2$ with $A_R,A_I\in\mf{H}(X)$. Assume that $\mcal{L}$ is defined on $\mf{B}(X)$ and injective when limited to $\mf{H}(X)$. If $\mcal{L}(A)=0$, then $\mcal{L}(A_R)+i\mcal{L}(A_I)=0$ and it follows that $\mcal{L}(A_R)=\mcal{L}(A_I)=0$, which in turn implies that $A_R=A_I=0$. Therefore $A=0$, from which we conclude that $\mcal{L}$ is injective on $\mf{B}(X)$.
    
    Next, assume that $\mcal{T} \in \mf{LH}(X,Y)$ and surjective when limited to $\mf{H}(X)$. It means that for every $B\in \mf{H}(Y)$ there exists $A\in\mf{H}(X)$ such that $\mcal{T}(A)=B$. Hence, for any $D=D_R+iD_I\in\mf{B}(Y)$, there exist $C_R$ and $C_I$ in $\mf{H}(X)$ such that $\mcal{T}(C_R)=D_R$ and $\mcal{T}(C_I)=D_I$. It follows that for $C$ given as $C=C_R+iC_I$, we have $\mcal{T}(C)=D$, thus $\mcal{T}$ is also surjective as a linear map from $\mf{B}(X)$ to $\mf{B}(Y)$.
    
\end{proof}

\begin{lemma} \label{applem:2}
    Let $\mcal{R}\in\mf{B}(X,Y)$. If $\mcal{R}\circ\mcal{N}=0$ implies $\mcal{N}=0$ for any $\mcal{N}\in\mf{B}(X)$, then $\mcal{R}$ is injective. Similarly, if $\mcal{N}\circ \mcal{R}=0$ implies $\mcal{N}=0$ for any $\mcal{N}\in\mf{B}(Y)$, then $\mcal{R}$ is surjective.
\end{lemma}
\begin{proof}
    For the first part, let us assume that $\mcal{R}$ is not injective, thus $\ker{\mcal{R}}\neq0$. Let $\mcal{N}$ be the projection onto $\ker{\mcal{R}}$. Then, even though $\mcal{N}$ is not trivial, $\mcal{R}\circ\mcal{N}=0$. For the second part, let us assume that $\mcal{R}$ is not surjective, i.e. $\Im{\mcal{R}}\neq Y$. In this case, if we let $\mcal{N}$ be the projection onto $\Im{\mcal{R}}^\perp$, then, even though $\mcal{N}\neq 0$, we have $\mcal{N}\circ\mcal{R}=0$.
\end{proof}

\begin{lemma} \label{applem:3}
    If a Hermitian-preserving linear map $\mcal{K}\in \mf{LH}(X,Y)$ such that $\Tr[M]=0$ implies $\Tr[\mcal{K}(M)]=0$ is injective (or surjective onto $\mf{H}_0(Y)$) when limited on $\mf{H}_0(X)$, then it is also injective (or surjective onto $\mf{H}(Y)$) when limited on $\mf{H}(X)$  if $\Tr[\mcal{K}(\mds{1}_X)]\neq0$.
\end{lemma}
\begin{proof}
    It follows immediately from the fact that $\mf{H}(X)=\mf{H}_0(X)\oplus\text{Span}_{\mds{R}}(\{\mds{1}_X\})$ and that $\Tr[\mcal{K}(\mds{1}_X)]\neq0$ implies $\mcal{K}(\mds{1}_X)\notin \mf{H}_0(Y)$.
\end{proof}

\subsection{Proof of Proposition \ref{thm:dariano}}

\begin{proof}
    We remark that $\mcal{R}\in\mf{C}(X,Y)\subseteq\mf{LH}(X,Y)$ and $\mcal{R}(\mf{H}_0(X))\subseteq\mf{H}_0(Y)$.  Now, from the assumption, we claim that $\mcal{R}\circ\mcal{K}=0$ implies $\mcal{K}=0$ for every $\mcal{K}\in\mf{B}(\mf{H}_0(X))$ (or $\mcal{P}\circ\mcal{R}=0$ implies $\mcal{P}=0$ for every $\mcal{P}\in\mf{B}(\mf{H}_0(Y))$) because of Lemma \ref{applem:1}. By Lemma \ref{applem:2}, it follows that $\mcal{R}$ is injective (or surjective onto $\mf{H}_0(Y)$) on $\mf{H}_0(X)$. And since $\mcal{R}\in\mf{C}(X,Y)$, we have $\Tr[\mcal{R}(\mds{1}_X)]=|X|\neq 0$, thus, from Lemma \ref{applem:3}, it follows that $\mcal{R}$ is injective (or surjective onto $\mf{H}(Y)$) on $\mf{H}(X)$.  Finally, by Lemma \ref{applem:4}, $\mcal{R}$ is also an injection (or a surjection) on $\mf{B}(X)$.
\end{proof}
\subsection{Proof of Corollary \ref{appco:1}}
\begin{proof}
   Since $\Tr[\mcal{R}(\rho)]=\Tr[\mcal{R}^\dag(\mds{1}_X)\rho]>0$ for any nonzero $\rho\geq0$, we have $\mcal{R}^\dag(\mds{1}_X)>0$, thus $\mcal{R}^\dag(\mds{1}_X)^{-1}$ exists. Let $\mcal{R'}:=\mcal{R}\circ\Ad_{\mcal{R}^\dag(\mds{1}_X)^{-1/2}}.$ Note that $\mcal{R}'$ is a quantum channel. Since $\Ad_{\mcal{R}^\dag(\mds{1}_X)^{-1/2}}$ is invertible, if $\mcal{N}\circ\mcal{R}=\mcal{M}\circ\mcal{R}$ implies $\mcal{N}=\mcal{M}$, then $\mcal{N}\circ\mcal{R'}=\mcal{M}\circ\mcal{R'}$ also implies $\mcal{N}=\mcal{M}$. It follows that $\mcal{R}'$ is surjective, hence $\mcal{R}$ is also surjective.
\end{proof}
\subsection{Proof of Lemma \ref{applem:5}}
\begin{proof}
    If follows from the observation that $\Im{\mcal{R}}^\perp=\ker{\mcal{R}^\dag}$. Since $\mcal{R}$ is surjective if and only if $\Im{\mcal{R}}=Y$ and $\mcal{R}^\dag$ is injective if and only if $\ker{\mcal{R}^\dag}=\{0\}$, the desired result follows. Note that the injectivity of $\mcal{R}^\dag$ and that of $\mcal{R}^T$ are equivalent since $\mcal{R}^\dag=(\mcal{R}^T)^*$ and the complex conjugation is bijective.
\end{proof}
\subsection{Proof of Theorem \ref{thm:equiv}}
\begin{proof}
    Note that by showing the equivalence of faithfulnesses to quantum channels and random unitary channels, one also shows the equivalence with that to unital channels because of the inclusion relation between those classes of quantum channels. Obviously a quantum state faithful to quantum channels is faithful to any subset of quantum channels. We now show the converse. Let
    \begin{equation}
        \text{zpan}_\mds{R}(\mcal{S}):=\left\{\sum_i \lambda_is_i:\lambda_i\in \mds{R}, \sum_i\lambda_i=0, \text{and } s_i \in \mcal{S}\right\}.
    \end{equation}
    It was shown in Ref. \cite{mendl2009unital} that $\text{zpan}_\mds{R}(\mf{UO}(X))=\mf{B}(\mf{H}_0(X)).$ Note that $\zpan{\mf{UO}(X)}=\{\alpha(\mcal{N}-\mcal{M}):\alpha>0\text{ and }\mcal{N},\mcal{M}\in\mf{RU}(X)\}$. By following the same argument of the proof of Proposition \ref{thm:dariano}, we can prove that $\mcal{R}\in\mf{C}(X,Y)$ is surjective if and only if $\mcal{N}\circ\mcal{R}=\mcal{M}\circ\mcal{R}$ implies $\mcal{N}=\mcal{M}$ for any $\mcal{N},\mcal{M}\in\mf{RU}(Y)$. Therefore, following the same logic with the proof of Theorem \ref{thm:dariano2}, it follows that a bipartite state $\rho_{AB}$ is faithful to random unitary operations if and only if its \jami map $\rho_{A\to B}$ is left invertible. In conclusion, faithfulness to quantum channels and faithfulness to random unitary operations are equivalent.
\end{proof}
\subsection{Proof of Proposition \ref{prop:unitinv}}
\begin{proof}
    Let $\rho_{AB}=1/2(\sum_{i=1}^d\lambda_i\dyad{i}_A)\otimes\dyad{0}_B+1/2\dyad{+}_A\otimes\dyad{1}_B$ in $\mf{S}(\mds{C}^d\otimes \mds{C}^2)$, where $\sum_{i=1}^d\lambda_i\dyad{i}_A$ is a non-degenerate quantum state, and $\ket{+}_A=|A|^{-1/2}\sum_{i=1}^d\ket{i}_A.$ One can see that $\rho_{AB}$ is faithful to unitary operations on $A$ as the pair $(\sum_{i=1}^d\lambda_i\dyad{i}_A,\dyad{+}_A)$ jointly transforms uniquely for each unitary operation. However, the corresponding \jami map $\rho_{B\to A}(\sigma)=1/2(\sum_i\lambda_{i=1}^d\dyad{i}_A)\bra{0}\sigma\ket{0}+1/2\dyad{+}\bra{1}\sigma\ket{1}$ cannot be surjective since its image is $\text{span}_\mds{C}\left\{\sum_i\lambda_{i=1}^d\dyad{i}_A, \dyad{+}_A\right\}$ while $\mf{B}(A)$ is $|A|^2$-dimensional complex vector space with $|A|\geq 2$. Thus, according to Lemma \ref{applem:5}, $\rho_{A\to B}$ cannot be injective.
\end{proof}
\subsection{Proof of Proposition \ref{prop:nonequiv}}
\begin{proof}
    Consider the same $\rho_{AB}$ of the proof of Proposition \ref{prop:unitinv}. We claim that this state is sensitive to unital channels on $A$ too. To show this, we first remark that if a unital channel $\mcal{N}\in\mf{UC}(A)$ preserves $\rho_{AB}$, then it preserves both $\sum_{i=1}^d\lambda_i\dyad{i}_A$ and $\dyad{+}_A$. But, if unital channel preserves a quantum state, then it also preserves the eigenspace corresponding to each eigenvalue \cite{burgarth2013ergodic,lie2020uniform,lie2021correlational}, thus $\mcal{N}$ also preserves every projector onto basis element $\dyad{i}$. It fits the definition of (generalized) dephasing maps, and it is known \cite{kye1995positive,li1997special} that they can be expressed as a Schur-product channel $(M_{ij})\mapsto (C_{ij}M_{ij})$ with a correlation matrix $C$ (meaning $C\geq0$ and $C_{ii}=1$ for every $i$). However, since $\mcal{N}$ must also preserve $\dyad{+}_A$, whose every matrix component is $|A|^{-1}$ so that it is mapped to $(C_{ij}|A|^{-1})$ by the Schur product with $C$, the corresponding correlation matrix should have 1 for every component. It is equivalent to saying that $\mcal{N}=\id_A$. Since every random unitary operation is unital, the same argument also holds for it too. Nonetheless, by Theorem \ref{thm:equiv} and Proposition \ref{prop:unitinv}, $\rho_{AB}$ is neither faithful to unital channels nor random unitary operations.
\end{proof}
\subsection{Proof of Theorem \ref{thm:PCQcond}}
We first prove the following Lemma.
\begin{lemma} \label{prop:sensun}
    A bipartite quantum state $\rho_{AB}$ is sensitive to unital channels on $A$ if and only if for any $M\in\mf{B}(A)$, $[M_A\otimes\mds{1}_B,\rho_{AB}]=0$ implies $M_A=\alpha\mds{1}_A$ for some complex number $\alpha$. The same result holds for unitary operations with any $M\in\mf{U}(A)$ instead of $\mf{B}(A)$.
\end{lemma}
\begin{proof}
    A quantum state is fixed by a unital channel if and only if it commutes with every Kraus operator of the quantum channel \cite{arias2002fixed}. Therefore, if $[M_A\otimes\mds{1}_B,\rho_{AB}]=0$ implies $M_A=\alpha\mds{1}_A$ for any $M\in\mf{B}(A)$ with some complex number $\alpha$, then $\mcal{E}_A(\rho_{AB})=\rho_{AB}$ if and only if $\mcal{E}_A=\id_A$ for any $\mcal{E}\in\mf{UC}(A)$ because it means that every Kraus operator of $\mcal{E}$ commutes with $\rho_{AB}$, thus is proportional to the identity operator. The same argument can be applied to unitary operations because they are also unital channels.
    
    Conversely, assume that $\rho_{AB}$ is sensitive to unital channels on $A$ and let $M\in\mf{B}(A)$ be an arbitrary matrix. Without loss of generality we assume $M\neq0$. If $\norm{M}_\infty>1$, then redefine $M\to M/\norm{M}_\infty$ so that $\norm{M}_\infty\leq1$. Suppose that $[M\otimes \mds{1}_B,\rho_{AB}]=0$. It implies both $[M_R\otimes \mds{1}_B,\rho_{AB}]=0$ and $[M_I\otimes \mds{1}_B,\rho_{AB}]=0$ because $\rho_{AB}$ is Hermitian (See the proof of Lemma \ref{applem:4}). Let $M'_k:=(\mds{1}_A-M_k^2)^{1/2}$ for $k=R,I$. Note that $[M'_k\otimes \mds{1}_B,\rho_{AB}]=0$ for $k=R,I$ since $M_k$ and $M'_k$ share the eigenbasis. After doing this, we can define unital channels $\mcal{I}_k:=\Ad_{M_k} + \Ad_{M'_k}$ for $k=R,I$ that preserve $\rho_{AB}$. By the sensitivity of $\rho_{AB}$, it follows that both of $\mcal{I}_k$ are the identity operation. It implies that $\Ad_{M_k}\propto \id_A$, hence $M_k\propto \mds{1}_A$ for both $k=R,I$ and consequently $M=\alpha \mds{1}_A$ for some $\alpha$.
    
    Similarly, assume that $\rho_{AB}$ is sensitive to unitary operations on $A$ and let $U\in\mf{U}(A)$ be an arbitrary unitary operator on $A$. Suppose that $[U_A\otimes \mds{1}_B,\rho_{AB}]=0$. Then, the unitary operation $\Ad_U$ on $A$ fixes $\rho_{AB}$, thus $\Ad_U=\id_A$ because of the sensitivity of $\rho_{AB}$. It implies that $U\propto \mds{1}_A$.
\end{proof}
Now we prove Theorem \ref{thm:PCQcond}.
\begin{proof}
    Here we claim that $\rho_{AB}$ is not a PC-Q state if and only if $[U_A\otimes\mds{1}_B,\rho_{AB}]=0$ implies $U_A=\alpha\mds{1}_A$ for any $U\in\mf{U}(A)$. First, if $\rho_{AB}$ satisfies (\ref{eqn:PCQ}), then unitary operator $\mds{1}_A-2\Pi_1$ commutes with $\rho_{AB}$, but it is not proportional to the identity operator. Conversely, if there is a unitary operator with the spectral decomposition $U=\sum_j \exp(i\theta_j)\Pi_j$ that  is not proportional to the identity operator and commutes with $\rho_{AB}$, then $[\Pi_i\otimes\mds{1}_B,\rho_{AB}]=0$ for all $i$, so (\ref{eqn:PCQ}) holds.
    
    Next, we prove that $\rho_{AB}$ is not a PC-Q state if and only if $[M_A\otimes\mds{1}_B,\rho_{AB}]=0$ implies $M_A=\alpha\mds{1}_A$ for any $M\in\mf{B}(A)$. Note that unitary operator is just a special case of general operator $(\mf{U}(A)\subseteq\mf{B}(A))$, hence one direction is already proved in the last paragraph. Conversely, assume that there exists $M\not\propto \mds{1}_A$ such that $[M_A\otimes\mds{1}_B,\rho_{AB}]=0$. It follows that at least one of $M_R$ and $M_I$ is not proportional to $\mds{1}_A$ and they both commute with $\rho_{AB}$, but they are Hermitian in contrast to $M$ itself. Therefore, at least one of $M_R$ and $M_I$ has the nontrivial spectral decomposition of the form of $\sum_i m_i \Pi_i$, and $\rho_{AB}$ commutes with every $\Pi_i$, i.e., $[\Pi_i\otimes\mds{1}_B,\rho_{AB}]=0$ for all $i$. It follows that (\ref{eqn:PCQ}) holds.
\end{proof}

\end{document}